\pdfoutput=1
\documentclass[12pt]{amsart}
\usepackage{listings,amsfonts}
\usepackage{epsfig}
\newcommand{\R}{{\mathbb R}}
\newcommand{\dv}{{\hat{d}}}
\newcommand{\chim}{{\bar\chi}}
\newcommand{\db}{{\bar{d}}}
\newcommand{\bd}{{\bar{\delta}}}
\newcommand{\dhat}{{\hat{\Delta}}}
\newcommand{\pmax}{{p_{\max}}}

\newcommand{\nmeas}{n_{{\rm meas}}}
\newcommand{\npcs}{n_{{\rm pcs}}}
\newcommand{\nrdc}{n_{{\rm rdc}}}
\newcommand{\Rnmeas}{\R^{\nmeas}}
\newtheorem{thm}{Theorem}[section]
\newtheorem{prp}[thm]{Proposition}

\begin{document}
\title[Looking for central tendencies]{Looking for central tendencies in the conformational freedom of proteins using NMR measurements}
\author[F. Clarelli]{Fabrizio Clarelli$^{(1)}$}
\author[L. Sgheri]{Luca Sgheri$^{(1,*)}$}

\address{$^{(1)}$Istituto per le Applicazioni del Calcolo (CNR), sede di Firenze, Via Madonna del Piano, 10, 50019 Sesto Fiorentino (FI), Italy}

\email{$^*$l.sgheri@iac.cnr.it}
\begin{abstract}
We study the conformational freedom of a protein made by two rigid domains connected by a flexible linker. The conformational freedom is represented as an unknown probability distribution on the space of allowed states. A new algorithm for the calculation of the Maximum Allowable Probability is proposed, which can be extended to any type of measurements. In this paper we use Pseudo Contact Shifts and Residual Dipolar Coupling. We reconstruct a single  central tendency in the distribution and discuss in depth the results.
\end{abstract}

\maketitle

\section{Introduction}
Flexibility is a key point in the functioning of macromolecules such as proteins \cite{FSW,HM}. One of the few techniques which permit to extract information about the conformational freedom of proteins in physiological condition is NMR spectroscopy. In the last decade a vast literature has flourished on this topic, we refer the reader to the two recent review papers \cite{SB,RSPL} for a general discussion on the available techniques.

Since the temporal scale of the fluctuations in the fold of a protein is several orders smaller than the time needed to take NMR measurements \cite{PM}, information about the flexibility can be only be recovered as a probability function on the set of allowed states. This set can be parametrized using for instance Cartesian coordinates of atoms involving measurements, dihedral angles of the backbone or Euler transformations determining the position of rigid protein domains.

The recovery of this probability distribution is an under-determined ill-posed problem independently of the chosen parameters. The number of constraints is in fact far too small to determine a unique solution, except in trivial cases. Without further assumptions, any set of values of the parameters which is compatible with the measurements may be seen as a solution, and in principle there is no way of telling if a solution is better than another.

The lack of uniqueness, combined with the heterogeneity of the NMR measurement scenarios, led to a plethora of different techniques with different acronyms all trying to determine the {\emph{"best"}} solution. The two cited review papers try to classify (each from a different point of view) these approaches.

From the mathematical point of view it may be observed that there are two limit cases for the solutions.

\smallskip\noindent
$\bullet$ A first approach is to find the solution which minimizes the additional hypotheses on the data, thus using either the Maximum Entropy Principle (MEP) \cite{J}, or the Kullback-Leibler divergence \cite{KL}, which is the relative entropy. The MEP solution maximizes the uncertainty on the data, so each feature shown by the MEP solution is relevant. On the other hand the MEP solution is in general a continuous probability distribution, so that a large number of states is normally needed in order to approximate it \cite{CCV}. The large number of variables involved raises not easy computational issues.

\smallskip\noindent
$\bullet$ A second approach is to find if the measurements carry some preference for certain states. The Maximal Allowable Probability (MAP) is the largest weight of a given state in a probability distribution which is a solution. As a function of the state, the MAP is not a solution but a sharp bound from above. The zones with a large MAP are the positions favoured by the data. On the other hand it is not possible to establish to what extent the true distribution shows these asymmetries. We can only say that the largest asymmetries should be in favour of the zones indicated by the MAP technique.

\smallskip
Both approaches are equally able to recover the unknown probability distribution in the extremal cases. When there is very little conformational freedom the physical situation can be thought as a series of oscillations around a central state. In this case the MEP solution tends to a Dirac function of the central state, while the MAP estimate tends to be $1$ for the central state and $0$ for the other states. On the other hand when there is a very large conformational freedom the MEP solution tends to the uniform distribution in the space metric and the spread of the MAP estimate is minimal.

In the in-between cases the two approaches diverge. The MEP solution is obtained as the solution with the minimal spread in the probability density. On the other hand the MAP estimate for each state is obtained via solutions with the largest possible spread between the probability of the estimated state and the probabilities of the states needed to complete the solution. Since the problem is underdetermined, both approaches are consistent with the data. They focus on different aspects of the problem and they are in some sense complementary. For a deeper discussion of this topic, we refer the reader to \cite{RSPL}.

\smallskip
In this paper we show a development of the MAP approach which permits the combination of different NMR constraints. The MAP approach has been inspired by \cite{BDGKLPPPZ}, and the rigorous mathematical definition of the MAP has been given in \cite{GLS}, though different names have been subsequently used for this same bound of the probability. A geometrical algorithm has been developed in \cite{LLPS} to calculate the MAP estimate when only Residual Dipolar Coupling (RDC) \cite{TFKP} are considered, using the linearity properties of these measurements. Residual Dipolar Coupling and Pseudo Contact Shift (PCS) \cite{KG}, which are frequently obtained together, have been used to analyse the conformational freedom of calmodulin in \cite{BGLPPSY}, using a complicated and time-consuming numerical procedure. The main difficulty is that PCS (as is the case for most sets of data) do not possess the linearity properties of RDC, which permits working on averaged tensors.

The Maximum Occurrence algorithm \cite{BGLPPPRS} uses a predetermined pool of conformers to calculate the maximal probability. The choice of a finite number of conformers simplifies the algorithm and reduces the time needed for the calculations. With this choice the positions which would cause physical violations of the atoms of the two domains may be directly eliminated from the sample.

The SES (Sparse Ensemble Selection) method has been developed in \cite{BCSSNT}, and is focused on recovering a small set of conformers with large probabilities. A recent paper compares the two approaches and shows the information content of RDC and PCS \cite{ABFLPRS}.

In this paper we extend the geometrical algorithm of \cite{LLPS} to the case of PCS. Indeed, since we drop the linearity requirement only fulfilled by paramagnetic RDC, the approach can be extended to any set of measurements.

\section{Theory}
We use the calmodulin measurement scenario \cite{BGLPPSY} as our test case.
Calmodulin is a protein made by two rigid domains (the N and C terminals) connected by a flexible linker, see figure~\ref{fig:pdb} in section~\ref{se:imple}. A paramagnetic ion may be inserted in the binding sites of the N terminal, which is also called the metal domain. We can then measure NMR data for atoms belonging to both the N and C terminals.

The RDC measurements \cite{TFKP} are defined by
\begin{equation}\label{eq:rdc1}
\delta_{rdc,j}=\frac{c_{rdc}}{\|P_j\|^5}P_j^t\chi P_j
\end{equation}
where $P_j=P_{j_1}-P_{j_2}$ is the vector connecting selected pairs $j_1$ and $j_2$ of chemically linked atoms and $c_{rdc}$ is a constant. The paramagnetic tensor $\chi$ is a symmetric and null-trace $3\times 3$ matrix, thus depending on $5$ coefficients. Since the atoms are chemically bound, their distance may be considered fixed, so that $\|P_j\|$ is constant, and the only dependence is on the orientation of the vector $P_j$. If the atoms belongs to the same domain as the metal, the tensor $\chi$ and the vector $P_j$ belong to the same rigid structure. The RDC of the metal domain can be used to fit the numerical values of the $\chi$ tensor using (\ref{eq:rdc1}).

The PCS measurements \cite{KG} are given by 
\begin{equation}\label{eq:pcs1}
\delta_{pcs,j}=\frac{c_{pcs}}{\|P_j\|^5}P_j^t\chi P_j,
\end{equation}
a formula very similar to (\ref{eq:rdc1}), with a different constant. In this case $P_j$ is however the vector connecting the metal and selected atoms $j$. When the atom belongs to the metal domain there is no difference between RDC and PCS. In fact the atom does not move with respect to the metal ion, so $P_j$ is fixed. Indeed RDC and PCS of atoms of the metal domain can be coupled to obtain a better fit for the paramagnetic tensor $\chi$ (and possibly for the location of the paramagnetic metal ion), see for instance \cite{LPS}. From now on we suppose that this is the case, so that the paramagnetic tensor is known and only RDC and PCS for atoms belonging to the C terminal are considered. Since the C terminal moves with respect to the metal ion, the NMR measurements are averages of different states of the molecule, so that we may speak about \emph{mean} PCS or RDC.

The RDC and PCS are in fact the average of the values obtained for different positions of the C terminal (also called conformers). Each conformer is identified by an Euler transformation $E\equiv(R,t)$, where $R$ is a rotation and $t$ a translation. Note however that the $P_j$ of formulas (\ref{eq:rdc1}) for RDC are difference of coordinates, so that the translations cancel, and we have
\begin{equation}\label{eq:eul1}
E(P_j)=R(P_j-t){\rm\ for\ PCS},\quad E(P_j)=RP_j{\rm\ for\ RDC}.
\end{equation}
Note also that $\|E(P_j)\|=\|P_j\|$ does not depend on $(R,t)$ in the case of RDC. Because of the linker we can always suppose $t_{min}\le \|t\|\le t_{max}$, so that the space of allowed Euler transformations is compact.

Let $D$ be the space of probability distributions on this compact space. Each $d\in D$ is identified by the probability density $p(R,t)\ge 0$, such that $\int_{R,t} p(R,t)dRdt=1$. Then
\begin{equation}\label{eq:mrdc1}
\bd_{rdc,j}=\frac{c_{rdc}}{\|P_j\|^5}\int_{R,t} p(R,t) (R P_j)^t\chi (R P_j)dRdt.
\end{equation}
Since $p$ does not depend on $t$ for the RDC, we have $p=p(R)$ and $\int_{R} p(R)dR=1$. Using the \emph{mean paramagnetic tensor}
\begin{equation}\label{eq:chim1}
\chim=\int_{R} p(R) R^t\chi R dRdt.
\end{equation}
equation (\ref{eq:mrdc1}) becomes:
\begin{equation}\label{eq:rdc3}
\bd_{rdc,j}=\frac{c_{rdc}}{\|P_j\|^5}P_j^t\chim P_j.
\end{equation}

The same technique cannot be used for PCS because of the term $E(P_j)=R_i(P_j-t_i)$, so that
\begin{equation}\label{eq:mpcs1}
\bd_{pcs,j}=c_{pcs} \int_{R,t}\frac{p(R,t)}{\|R(P_j-t)\|^5} 
(R (P_j-t))^t\chi (R (P_j-t))dRdt.
\end{equation}

Different metal ions $M_k$ may be substituted in the same binding site belonging to the N terminal without influencing the fold of the protein \cite{ABJLLL}. We suppose that each set of measurements relative to metal $M_k$ is obtained by averaging values relative to conformers, using the same probability distribution $d\in D$. Note the following proposition, see for instance \cite{RB}.

\begin{prp}\label{prp:maxmet}
Independent PCS and RDC measurements may be obtained from at most $5$ different metal ions $M_k$.
\end{prp}
This is due to the fact there are at most $5$ linearly independent paramagnetic tensors $\chi^k$ relative to metals $k$. A sixth tensor $\chi^6$ may be written as a linear combination of the first five tensors. Hence PCS and RDC (and indeed mean PCS and RDC) relative to this sixth metal can be written as the same linear combination of the measurements relative to the first five metals, see (\ref{eq:mrdc1}) and (\ref{eq:mpcs1}).

In general we cannot determine $d$ from the measurements of the moving terminal. The problem is in fact underdetermined. The target distribution $d$ is a function of six variables, those defining the Euler transformation. If we only consider RDC, $d$ is a function of the three variables identifying the rotation, be them unitary quaternions or Euler angles.

On the other hand we only have a finite number of measurements. Moreover, it is well known that the maximal number of independent RDC measurements from atoms of the C terminal is 25, see for instance \cite{MPPGB}. Also, the information content of PCS is weak \cite{ABFLPRS}. Hence, no matter how many measurements are available, the distribution $d$ cannot be recovered except in some trivial cases.

\smallskip
We now report some well known properties of RDC and PCS, see for instance \cite{RB,AVTZEM}. We first examine the case of the RDC measurements.

\smallskip\noindent
\begin{prp}\label{prp:maxrdc}
The maximal number of independent mean RDC measurements is $25$.
\end{prp}
\noindent
This result is the consequence of two different properties:
\begin{itemize}
\item[(i)]The maximal number of independent mean paramagnetic tensors is $5$ (see Proposition~\ref{prp:maxmet}).
\item[(ii)]The maximal number of independent mean RDC for each metal is $5$.
\end{itemize}
For a proof, see \cite[Theorem 3.2]{LLPS}.

\noindent
More precisely, the independence of the RDC measurements is directly correlated to the independence of the mean paramagnetic tensors (\ref{eq:measm1}).
\begin{prp} Let $n$ be the number of independent mean paramagnetic tensors. Then the number of independent RDC measurements is $5n$.
\end{prp}

\noindent
For the proof we refer to \cite{LSS}.

Proposition~\ref{prp:maxmet} holds also for PCS, thus the maximum number of independent metal ions is again $5$. However, in principle each mean PCS is independent from the other, see formula (\ref{eq:measmcont}). In practise if two atoms $j_1$ and $j_2$ are close, so are $P_{j_1}$ and $P_{j_2}$. Hence the values of the averaged PCS from (\ref{eq:mpcs1}) are also close. Mathematically speaking we may observe that the values $\bd_{pcs,{j_1}}$ and $\bd_{pcs,{j_2}}$ are heavily correlated, so that the new information added by atom $j_2$ is very weak.

\section{The simplex algorithm}
\subsection{Geometrical setting}\label{subs:geo}
Suppose we have $n\le 5$ metal ions, and that $\chi^k$ are already given or determined via the RDC and PCS of the metal domain. Take any $d\in D$, we can calculate the mean RDC and PCS with the general formula
\begin{equation}\label{eq:measmcont}
\bd_j=c_j \int_{R,t}\frac{p(R,t)}{\|E(P_j)\|^5} 
E(P_j)^t\chi^{k_j} E(P_j)dRdt,
\end{equation}
where $p(R,t)$ is the probability density of $d$ at $(R,t)$, and $E(P)$ is defined by (\ref{eq:eul1}). The values $P_j$, $c_j$ and $k_j\le n$ depend on the choice of atoms and the type of measurement. The term $\|E(P_j)\|$ is constant for RDC. In the case of PCS, for physical reasons the distance between the metal ion and any other atom is anyway bounded away from $0$. Hence we may suppose that the measurements $|\bd_j|$ are uniformly bounded. We can obtain a certain number of RDC and PCS measurements for each of the $n$ metals, not necessarily referring to the same atoms. Let $\nrdc$ be the total number of mean RDC, $\npcs$ be the total number of mean PCS, and let $\nmeas=\nrdc+\npcs$.

We can collect the measurements calculated from (\ref{eq:measm1}) in a vector, so that each $d\in D$ defines a point $\bar\delta\in\R^{\nmeas}$. The key point of the geometrical approach is the projection from the space of finite distributions to the space of the measurements. Let $\Pi$ be such a projection, we may also decompose $\Pi$ into the RDC and PCS components:
\begin{equation}\label{eq:proj}
\Pi(d)\equiv\left(\begin{array}{c}\Pi_{rdc}(d)\\\Pi_{pcs}(d)\end{array}\right)=\left(\begin{array}{c}\bd_{rdc,1}\\\ldots\\
\bd_{rdc,\nrdc}\\\bd_{pcs,1}\\\ldots\\\bd_{pcs,\npcs}\end{array}\right).
\end{equation}
Let
\begin{equation}\label{eq:v2}
V=\{v\in\R^{\nmeas} : v=\Pi(d), d\in D\}.
\end{equation}
The set $V$ is compact because the measurements are uniformly bounded. The set $V$ is convex because if $v_i\in V$, $v_i=\Pi(d_i)$, $i=1,2$, then 
$\lambda d_1 + (1-\lambda)d_2\in D$, $\forall\lambda\in[0,1]$, so that
\begin{equation}
\lambda v_1 +(1-\lambda)v_2=\Pi(\lambda d_1 + (1-\lambda)d_2)\in V.
\end{equation}

Each convex set is the convex hull of its extreme points (also called vertices), i.e. the points that cannot be reconstructed using a convex combination of different points of the set. Let $\Delta\subset D$ the set of finite probability distributions, and let $\dhat\subset \Delta$ the set of probability distributions made by a single point. Because of the non-linearity, in general it is not true that each $\Pi(\dv)$ is a vertex of $V$, though we suspect this is the case in our setting. On the other hand, the set of vertices is a subset of $\Pi(\dv)$, since $V$ is the set of convex combinations of these points. We do not need the property that each $\Pi(\dv)$ may be uniquely reconstructed, so we can nevertheless identify the set of vertices with $\Pi(\dv)$.

\begin{prp}\label{prp:finite}For each $d\in D$ there exists a $\tilde d\in \Delta$ such that $\Pi(\tilde d)=\Pi(d)$.
\end{prp}

\begin{proof}By Carath\'{e}odory's theorem, each $\Pi(d)\in V$ may be reconstructed with a convex combination of at most $\nmeas+1$ vertices of $V$. Let $\Pi(d)=\sum_{i=0}^{\nmeas} p_i \Pi(\dv_i)$, with $\dv_i\in\dhat$.
Then $\tilde d=\sum_{i=0}^{\nmeas} p_i \dv_i\in \Delta$ is the required distribution.
\end{proof}

\smallskip\noindent
{\bf{Remark:\ }}Proposition~\ref{prp:finite} entitles us to work with finite distributions of probability without loss of generality. If $d\equiv(p_i,R_i,t_i)\in \Delta$, formula (\ref{eq:measmcont}) may be rewritten as

\begin{equation}\label{eq:measm1}
\bd_j=c_j \sum_i\frac{1}{\|E(P_j)\|^5} 
p_i E(P_j)^t\chi^{k_j} E(P_j).
\end{equation}

\begin{prp}\label{prp:m0}There exists a $d\in\Delta$ such that $\Pi(d)=0$.
\end{prp}

\begin{proof}The proposition is proven in \cite{GLS} for the case of RDC, and a constructive example with a finite distribution is given in \cite{S1}. Fix the origin of the Cartesian system in the binding site of the metal. Let $\tilde d\in D$ such that the translation $t$ is fixed and the rotational part coincides with the Haar measure $H(R)$, see for instance \cite{W}. Then $\Pi_{rdc}(\tilde d)=0$.  With these choices we also have $\Pi_{pcs}(\tilde d)=0$. Fix a $j$ relative to a PCS in formula (\ref{eq:measm1}). Let $\tilde P_j=P_j-t$, then $\|E(P_j)\|=\|R\tilde P_j\|=\|\tilde P_j\|$ for every rotation $R$ because the metal is in the origin. Hence
\begin{equation}\label{eq:measm2}
\begin{aligned}
\bd_j=&c_j \frac{1}{\|\tilde P_j\|^5} \int_R
(R\tilde P_j)^t\chi^{k_j} (R \tilde P_j)H(R)dR\\
=&c_j \frac{1}{\|\tilde P_j\|^5} \tilde P_j^t \left(\int_R R^t\chi^{k_j} R H(R)dR\right) \tilde P_j=0.
\end{aligned}
\end{equation}
This is due to the fact that the integrand in parenthesis is the mean paramagnetic tensor, and its integral is $0$ for the Haar measure \cite{GLS}. The existence of a $d\in\Delta$ is then guaranteed by Proposition~\ref{prp:finite}.
\end{proof}

The dimension $N\le \nmeas$ of the set $V$ is a key point which can be determined from the data. Using the results of the previous section, if we suppose that we have at least $5$ independent RDC measurements for each of the $n$ metal ions, then $N=5n+\npcs$. However, since the PCS are only marginally linearly independent, it is to be expected that there are directions where the set $V$ is very thin, so that the effective determination of $N$ should involve also some numerical considerations. In the supplementary information we analyse in detail the linear independence of the PCS versus the RDC measurements, and the consequences on the expected results.

\subsection{Definition of the MAP}
Let $\db$ the true unknown distribution of probability. Then, given any Euler transformation $(R,t)$ we define 
\begin{equation}\label{eq:pmax1}
\pmax(R,t)=\max_{d\in\Delta}\{p: (p,R,t)\in d {\rm\ and\ } \Pi(d)=\Pi(\db)\}. 
\end{equation}
In other words given any conformer, identified by the Euler transformation $(R,t)$, we look for the maximal coefficient $p$ that we can apply to this conformer in a convex combination such that the projection in $V$ is the same as that of $\db$. Suppose $\Pi(\db)$ belongs to the interior of $V$. Let $\Pi(\dv)=\Pi(1,R,t)$ be the vertex corresponding to the position $(R,t)$. Consider the line passing through $\Pi(\dv)$ and $\Pi(\db)$, the segment in between the two points belongs to $V$ because of the convexity. Moreover, since $\Pi(\db)$ is internal, there exists a point $\Pi(q)\in V$ on the continuation of the segment on the side of $\Pi(\db)$. Then $\Pi(\db)$ is the convex combination of $\Pi(q)$ and $\Pi(\dv)$, i.e. there exists a $p\in(0,1)$ such that
\begin{equation}\label{eq:dm1}
\Pi(\db)=p\Pi(\dv)+(1-p)\Pi(q).
\end{equation}
By definition we have $\pmax(R,t)\ge p$. The value $p$ can be explicitly determined using the distances (i.e. the $L^2$ norms) in $\R^N$, in fact
\begin{equation}\label{eq:dm2}
\Pi(\db)=\frac{\|\Pi(\db)-\Pi(q)\|}{\|\Pi(\dv)-\Pi(q)\|}\Pi(\dv)+\frac{\|\Pi(\db)-\Pi(\dv)\|}{\|\Pi(\dv)-\Pi(q)\|}\Pi(q).
\end{equation}
The maximal $p$ which verifies (\ref{eq:dm1}) is then obtained from
the $q$ with projection in $V$ having the maximal distance from $\Pi(\db)$. Because of the convexity, $\Pi(q)$ is the point on the boundary of $V$ on the continuation of the segment connecting $\Pi(\dv)$ and $\Pi(\db)$. Unfortunately, except in some trivial cases, there is no analytical procedure for determining if a point $\Pi(q)$ belongs to the boundary of $V$, so that we have to use an iterative procedure.

\subsection{The simplex algorithm}
Let $N\le 5n+\npcs$ be the dimension of $V$. By Carath\'{e}odory's theorem there are $N+1$ vertices of the convex $V$ such that
\begin{equation}\label{eq:sim1}
\Pi(\db)=\sum_{j=0}^{N} p_i\Pi(\dv_i^0),
\end{equation}
with $p_i\ge 0$ and $\sum_i p_i=1$. Note again that we cannot suppose that $\db=\sum_i p_i \dv_i^0$ because in general the solution is not unique, we can only recover the projection. Let $S_0\subset V$ be the simplex formed by the convex combinations of the vertices $\Pi(\dv_i^0)$. We may suppose $S_0$ is a simplex in $\R^N$, i.e. the vectors $\Pi(\dv_i^0)-\Pi(\dv_0^0)$ are linearly independent in $\R^N$. Since the set $\Pi(\dv)$ is connected we may choose $S_0$ so that $\Pi(\db)$ is internal to $S_0$, i.e. $p_i>0$ $\forall i$ \cite{LLPS}.

\begin{figure}[ht]
\centerline{\includegraphics[width=0.8\columnwidth]{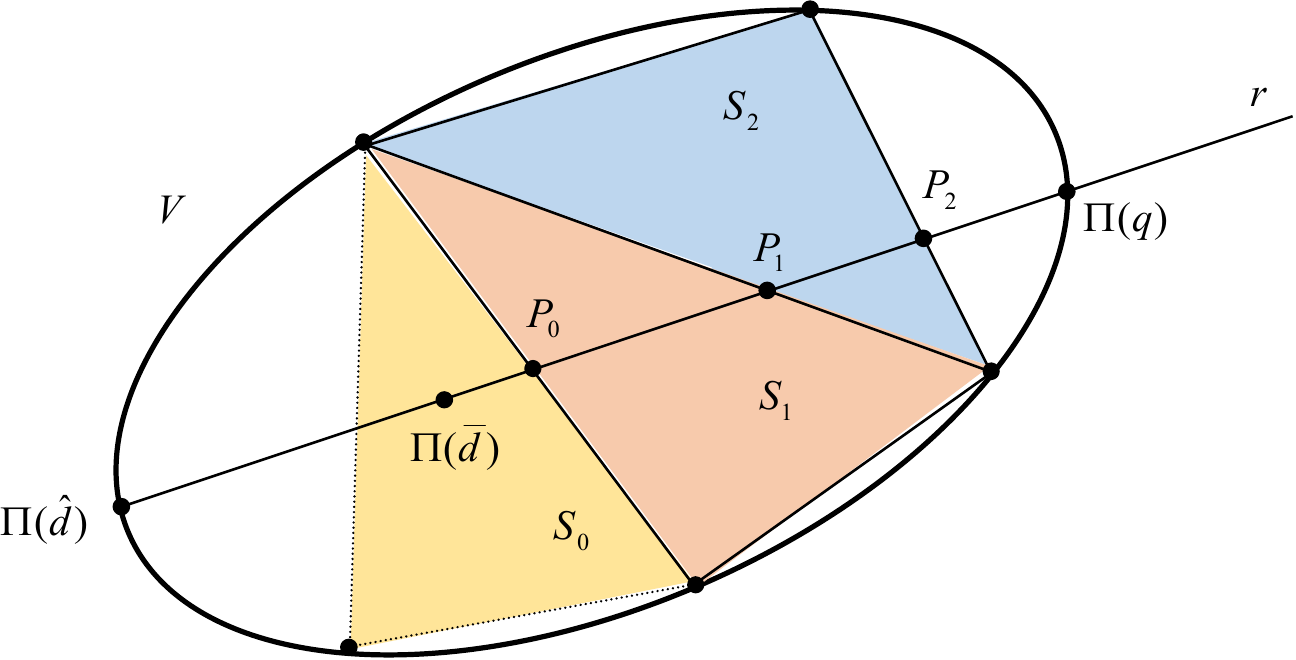}}
\caption{The simplex algorithm.\label{fig:simplex}}
\end{figure}


Now take any position $(R,t)$ and let $\hat d=(1,R,t)$, look at Fig.~\ref{fig:simplex} for reference. Take the line $r$ through $\Pi(\hat d)$ and $\Pi(\db)$. Since $\Pi(\db)$ is internal to $S_0$ there is a point $P_0\in\partial S_0$ on $r$ on the side opposite to $\Pi(\hat d)$ with respect to $\Pi(\db)$. The point $P_0$ identifies a face $F_{j_0}\subset S_0$ such that $P_0\in F_{j_0}$. The face $F_{j_0}$ is identified by removing the vertex ${j_0}$ from the set of vertices of $S_0$.

The point $P_0$ is either on the boundary or internal to $V$. In the first case we are finished because we have found the point needed by the definition of $\pmax$. In the second case, consider the hyperplane $H_{j_0}$ containing the face $F_{j_0}$. The hyperplane $H_{j_0}$ cannot be a support hyperplane since it contains an internal point, thus there will be at least a vertex $\dv_{j_0}^1$ on the half-plane defined by $H_{j_0}$ and not containing $\Pi(\db)$. Define $S_1$ to be the simplex with $\dv_{j_0}^0$ replaced by $\dv_{j_0}^1$.

We can iterate the algorithm, each time finding the two intersections of $r$ with the simplex $S_k$. The intersection point $P_k$ on $r$ farther from $\Pi(\db)$ determines a face $F_{j_k}$ of the simplex $S_k$. If this intersection point is internal we can replace the vertex $j_k$ of $S_k$ not belonging to $F_{j_k}$ with a new one, lying in the half-space determined by the hyperplane $H_{j_k}$ and not containing $\Pi(\db)$.

Thus we determine a monotonic sequence of points $P_k\in r$ converging to a point $P$. The point $P$ cannot be internal to $V$, otherwise the algorithm would have found a new replacement vertex. Then $P\in \partial V$is the point needed by the definition of $\pmax$.

\subsection{Determination of the projection matrix}
In principle the algorithm may be carried out in the ambient space $\R^{\nmeas}$ without any modifications. However, the dimension $N$ of $V$ is in general strictly smaller than the number of measurements $\nmeas$. The simplex algorithm works in the linear subspace spanned by $V$, which has dimension $N$. Using the $\nmeas$ ambient coordinates in this linear subspace is definitely a bad idea, because any numerical approximation in the calculations is likely to bring the points out of the linear subspace. Thus the first step is to determine the dimension $N$ of $V$, and the projection operator from $\Rnmeas$ into $\R^N$.

The dimension $N$ is the maximal number of linearly independent vectors of the form $\Pi(d_i)-\Pi(d_0)$, where $d_0$ is a fixed point in $\Delta$, and $d_i\in \Delta$. Because of Proposition~\ref{prp:m0}, we may take $d_0$ such that $\Pi(d_0)=0$. Since each point in $V$ is a convex combination of vertices, $N$ is then the maximal number of linearly independent vectors $\Pi(\dv_i)$, where $\dv_i\in\dhat$. 


The Singular Value Decomposition (SVD, see for instance \cite{PTVF}) may be used to determine $N$, as already done in \cite{S2} in a different context. Take points $\Pi(\dv_i)\in V$, $i=1,\ldots,M$, with $M>>\nmeas$, and form the matrix
\begin{equation}\label{eq:svd1}
A=(\Pi(\dv_1),\ldots,\Pi(\dv_M)),
\end{equation}
which has dimension $\nmeas\times M$. The SVD is based on the singular values of $A$, which are the square roots of the eigenvalues of the symmetric and positive semi-definite matrix $A^tA$.

The SVD decomposes the matrix $A$ in the form
\begin{equation}\label{eq:svd2}
A=U\Lambda W, 
\end{equation}
where $\Lambda$ is a $\nmeas\times \nmeas$ diagonal matrix containing the singular values of $A$ in decreasing order, $W$ is an orthogonal $\nmeas\times \nmeas$ matrix, and $U$ is a $M\times \nmeas$, column-orthogonal matrix, i.e. $\sum_k u_{ki}u_{kj}=\delta_{ij}$. Since the rank of $A$ is by definition $N$, the matrix $A$ has exactly $N$ non-zero singular values.

The matrix $W$ can be used as a projection matrix. In fact, for any $v\in V$ we have that $W v$ is the decomposition of $v$ with respect to the base of eigenvectors of $A^tA$. Moreover, since the matrix $\Lambda$ has only $N$ non-zero eigenvalues, we are only interested in the matrix $W_N\subset W$, including only the first $N$ rows of $W$. 


While the matrix $W_N^t W_N$ is not the identity, it works as such on the points of $V$. In fact the points of the linear space spanned by $V$ can be uniquely identified either by a subset of the $\nmeas$ coordinates satisfying $\nmeas-N$ linear conditions, or by the $N$ intrinsic coordinates obtained by applying $W_N$.

The following diagram pictures the situation:
{
\begin{equation*}
\begin{array}{cccccccccc}
\Delta & & \R^{\nmeas} & & \R^N & & \R^{\nmeas}\\
 & \Pi & & W_N & & W^t_N & \\
d & \rightarrow & \Pi(d) & \rightarrow & W_N(\Pi(d)) & \rightarrow & \Pi(d) \\
\end{array}
\end{equation*}
} 

\section{Implementation and self-consistency tests}\label{se:imple}
In this section we describe the implementation of the simplex algorithm and report the results of simulations run with exact measurements. It is a "proof of concept" that the method works, and a necessary step to understand the interactions of the measurements before considering noisy data.

\subsection{Experimental setting}
As a model for calmodulin we use the pdb fold as determined by \cite{BGKLP}, shown in Figure~\ref{fig:pdb}.
\begin{figure}[ht]
\centerline{\includegraphics[width=0.8\columnwidth,trim={130 480 120 80},clip]{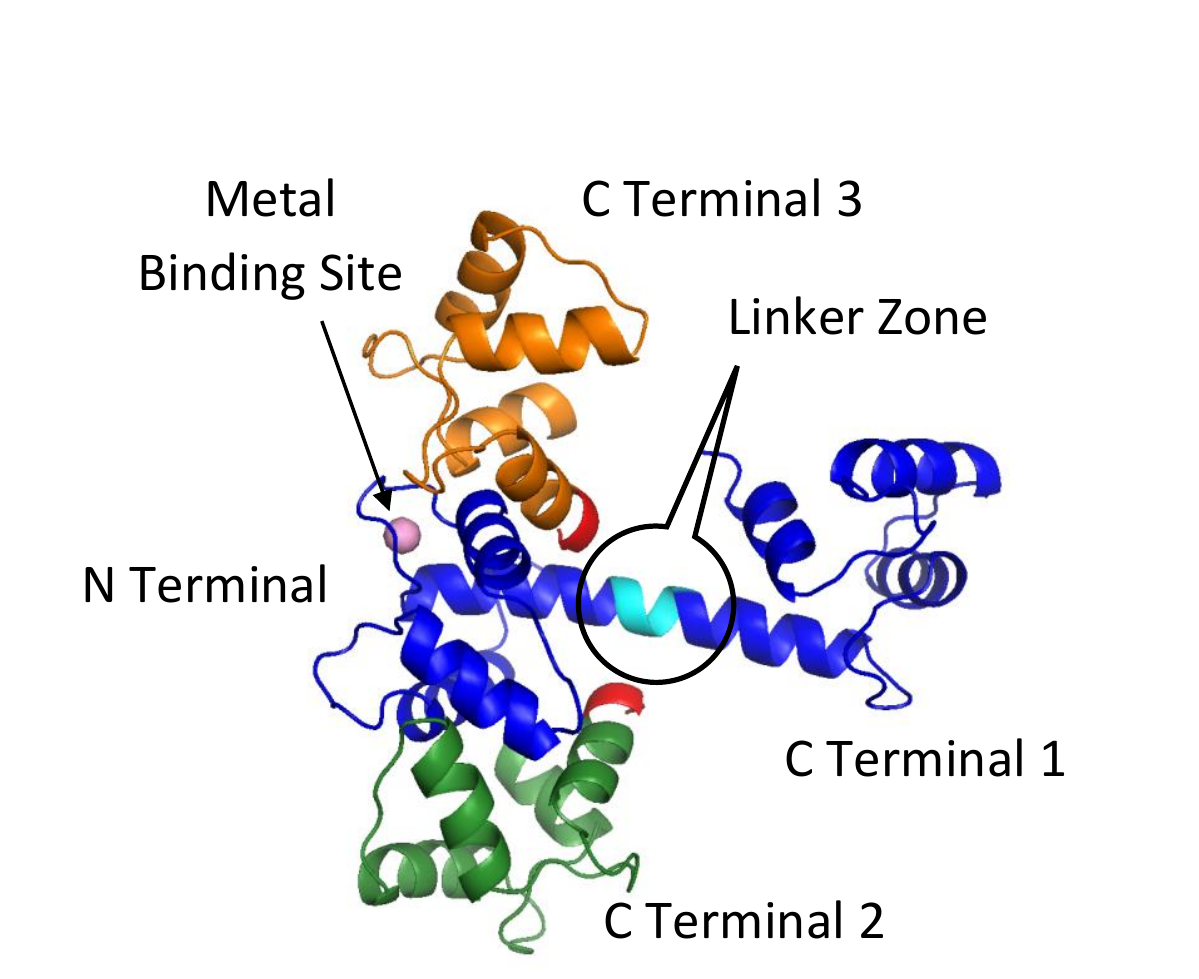}}
\caption{The fold of calmodulin and its conformational freedom. For the figure in color: \emph{pink:} metal position; \emph{blue:} N terminal and position $C_1$ of the C terminal, linker (\emph{cyan}) shown only for this position; \emph{green:} position $C_2$;  \emph{gold:} position $C_3$. Residuals 82--83 of C terminal shown in \emph{red} for $C_2$ and $C_3$ in order to highlight the connection points of the linker(not shown) to the C terminal.\label{fig:pdb}}
\end{figure}
In physiological conditions the long $\alpha$-helix breaks in the zone between residuals $77$ and $81$ \cite{BIKPB,BST}, resulting in some conformational freedom of the C terminal. Any position such that the linker length is between $6$\AA\ and $12$\AA\ is considered to be attainable. Any position such that there exists two ${\rm C}_\alpha$ atoms with a distance smaller than $3$\AA\ is considered to be a physical violation. A conformer satisfying both conditions is an allowable state for the C terminal. The positions of the C terminal shown are marked $C_1$ through $C_3$, with $C_1$ being the position where the linker remains folded into the $\alpha$-helix. In position $C_2$ the C terminal is close to the N terminal, but far from the metal. In position $C_3$ the C terminal is close both to the N terminal and to the metal. We have highlighted in red the first two residuals of the C terminal for positions $C_2$ and $C_3$ to show the connection point to the linker.

The measurements are generated using a continuous probability distribution $d_l\in D$ centered in conformers $C_l\equiv(R_l,t_l)$, $l=1,2,3$. The measurements are calculated using formula (\ref{eq:measm1}) as the arithmetic average of a large number $M$ of allowable states drawn according to the distribution $d_l$. More precisely, given two positive numbers $\sigma_R$ and $\sigma_t$ we draw conformers in the following way. The translation is drawn according to a Gaussian distribution with average $t_k$ and standard deviation of the module $\sigma_t$. The rotation is drawn according to a von Mises-Fisher distribution (see for instance \cite{W}) with average $R_k$ and standard deviation $\sigma_R$ of the rotational distance, calculated using quaternions. We only retain allowable states. The number $M$ is large enough to stabilize the measurements, thus simulating a continuous probability distribution. Loosely speaking the C terminal moves around the center position $C_k$ in such a way that the average deviation from the central position is $\sigma_t$ Amstrongs for the translation, and $\sigma_R$ degrees for the rotation. In this paper we use the numerical values $\sigma_t=3$\AA\ and $\sigma_R=20^\circ$.

We generated mean measurements with respect to three different paramagnetic tensors, corresponding to Tb, Tm and Dy lanthanide ions substituted for Ca in the second binding site of the N-terminal. We simulated a total of $112$ mean RDC using N-H dipoles from residuals of the C terminals, and a total of $132$ mean PCS using HN atoms from the C terminal.

In principle the distribution is symmetric around the center. However the constraint on the physical violations may introduce asymmetries in the distribution. This happens in cases $C_2$ and $C_3$, when the center position of the C terminal is close to the N terminal. As a consequence there is a small shift in the most probable position of the distribution. In the supplementary information we discuss in detail this issue.

\subsection{Observability of a central tendency}
Suppose there is a central tendency in the data. The MAP estimate is able to detect whether this tendency exists. To show this fact, we considered a probability distribution $d_0$ where all the allowable states are equally probable in the correct metric. If there were no constraints on the conformers the simulated measurements would be $0$ by Property~\ref{prp:m0}. The large conformational freedom is however detectable from the RDC measurements. The standard deviation of the simulated RDC is in fact $0.36$ for $d_0$, while is larger than $5$ for the $d_l$ cases, $l\ge 1$. The situation is different for PCS. In this case, small values are obtained both when there is a large conformational freedom and when the C terminal is far from the binding site of the metal. The standard deviation of the simulated PCS is $0.10$ for $d_0$, $0.14$ for $d_1$, $0.08$ for $d_2$ and finally $0.94$ for $d_3$, the case where the C terminal is closer to the metal binding site.

The MAP estimate detects this difference in the data. In the $d_0$ case using RDC we found $0.31\le \pmax(R) \le 0.34$ for all the orientations of the C terminal. Typical values for the other $d_i$ cases are $0.1\le \pmax(R) \le 0.70$, thus having a much larger span. The MAP estimate is then in principle able to detect an asymmetry in the data due to a restricted conformational freedom.

\subsection{Determination of the central tendency}\label{ss:algorithm}
We used the following steps to determine a central tendency of the measurement.
\begin{itemize}
\item[1.\ ]The orientation $R_0$ of the conformer with the largest $\pmax$ is determined using RDC alone.
\item[2.\ ]The translation $t_0$ of the conformer with the largest $\pmax$ is determined using PCS alone. The rotation is kept fixed at $R=R_0$.
\end{itemize}

In the following figures we present the results of the tests. 

\begin{figure}[ht]
\centerline{\includegraphics[width=\columnwidth, trim={20 40 40 30}]{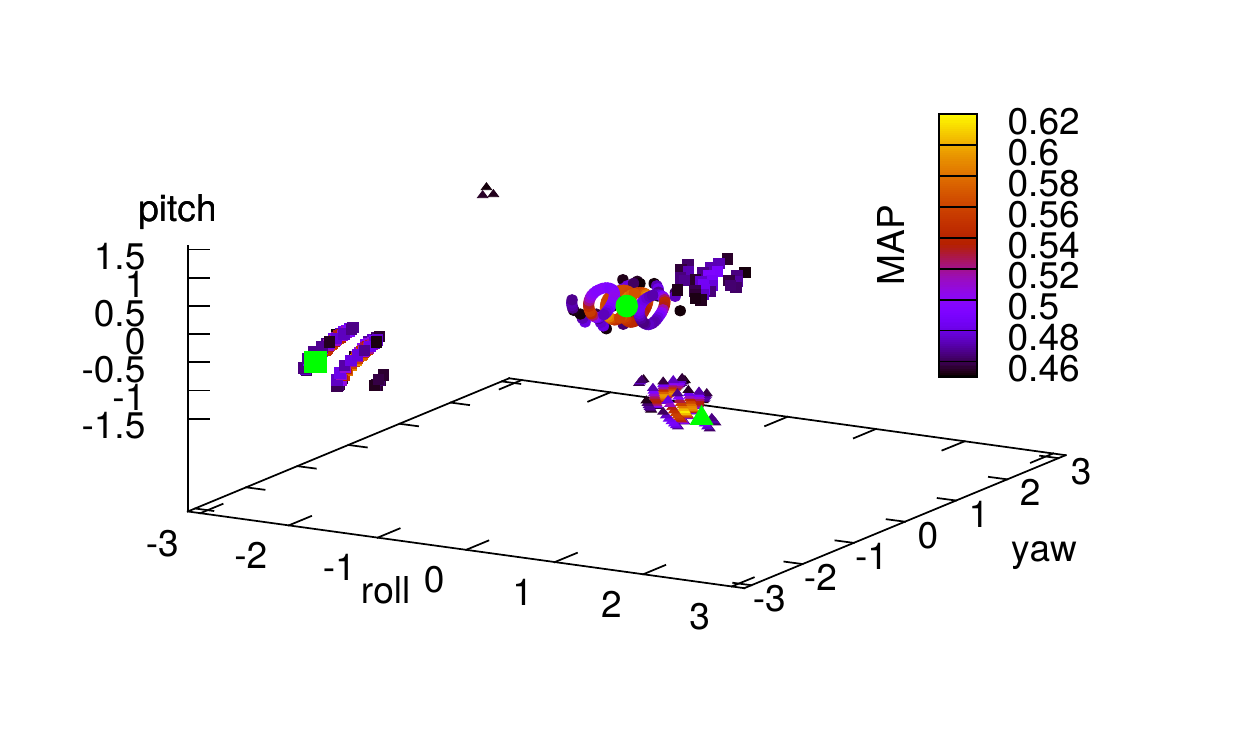}}
\caption{Tests with RDC alone. Cases $C_1$ (circles), $C_2$ (squares), $C_3$ (triangles).\label{fig:test1noerr}}
\end{figure}

In Figure~\ref{fig:test1noerr} we show the results of step 1 of the algorithm in cases $C_1$ (circles), $C_2$ (squares), $C_3$ (triangles). The points represent the orientations of the sample with MAP larger than a certain threshold. The larger green dots, here and in the subsequent figures, mark the positions of the centers. In cases $C_2$ and $C_3$ there are two different zones with large MAP, due to the so-called phenomenon of \emph{ghost cones} \cite{LLPS}, which derives from the symmetries of the RDC formula.

\begin{figure}[ht]
\centerline{
\includegraphics[width=0.49\columnwidth, trim={20 40 30 30}]{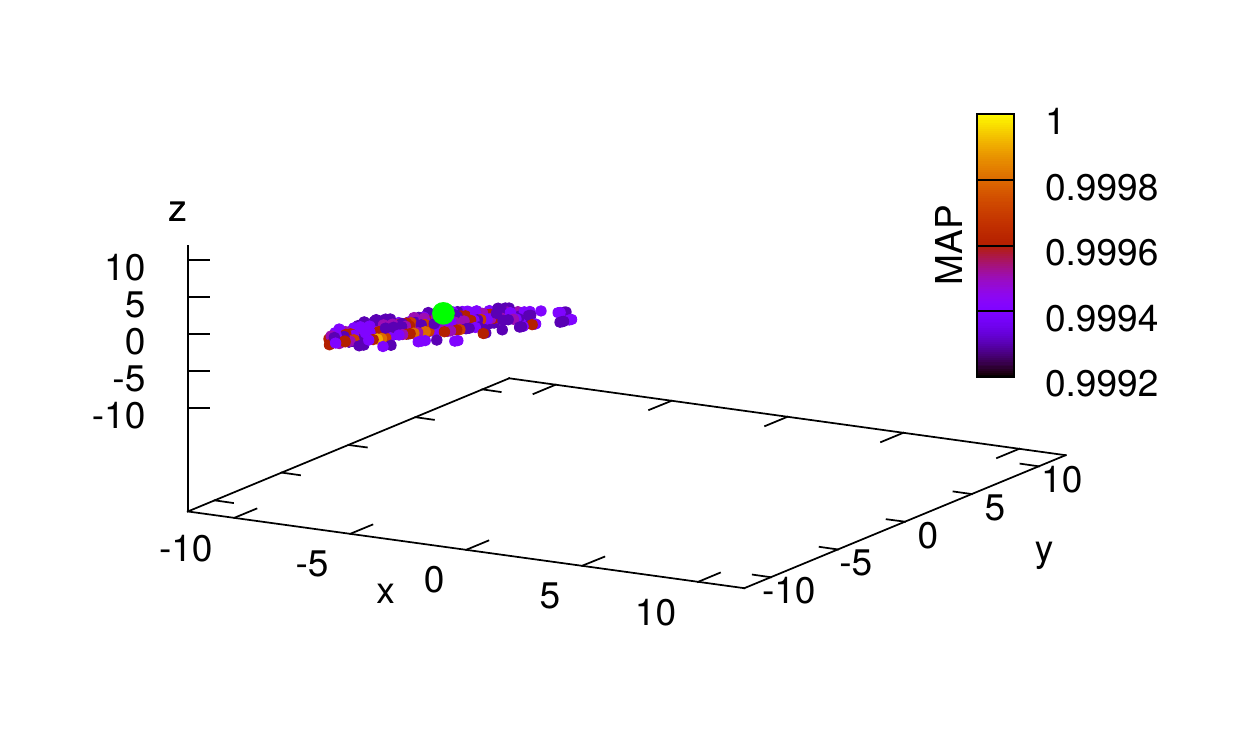}
\includegraphics[width=0.49\columnwidth, trim={10 40 40 30}]{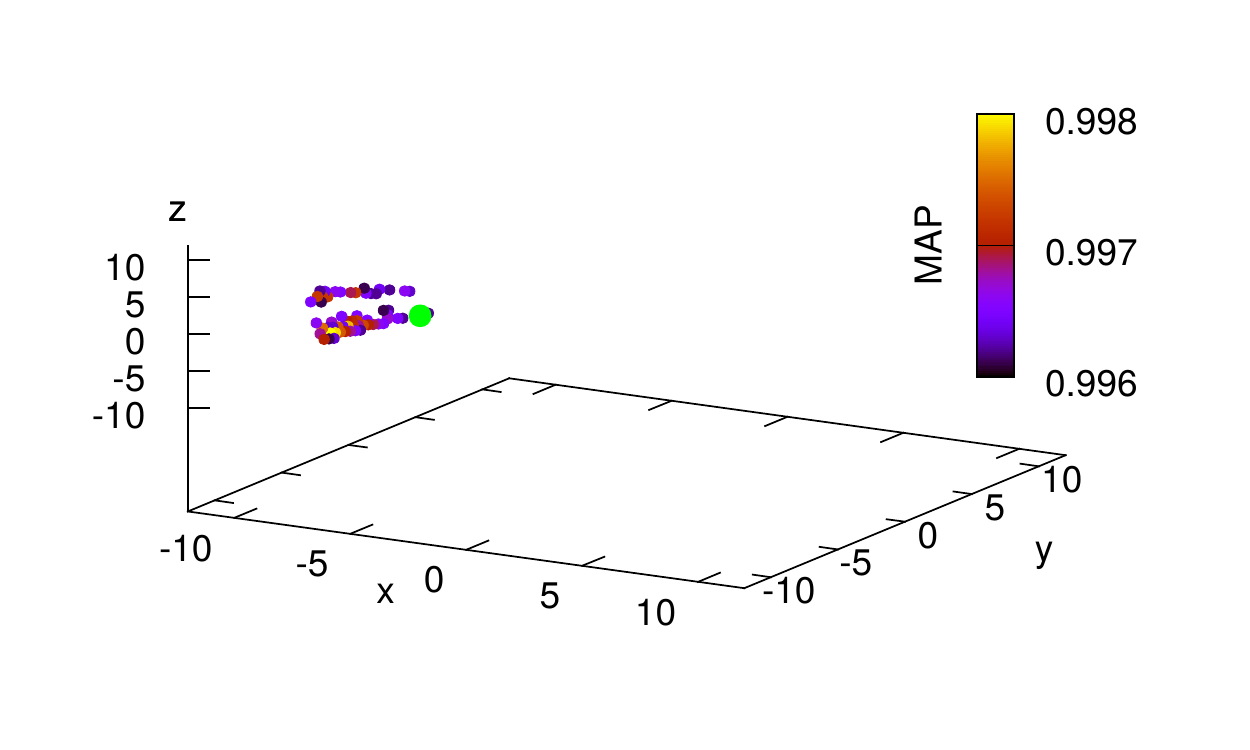}}
\caption{Tests with PCS alone. Cases $C_1$ (left panel) and $C_3$ (right panel).\label{fig:test2noerr}}
\end{figure}

In Figure~\ref{fig:test2noerr} we show the results of step 2 of the algorithm in cases $C_1$ and $C_3$, case $C_2$ being very similar to case $C_1$. Note that the MAP of a large number of translations is close to $1$, as a consequence of the poor resolving power of the PCS. A slightly better reconstruction is obtained for $C_3$. In this case the center position is close to the metal, so that there are some PCS values which are rather large, see formula~(\ref{eq:mpcs1}). To obtain these values the C terminal must remain close to the $C_3$ position for a not negligible fraction of time. As a consequence, the $\pmax$ values for positions far from the metal is slightly reduced.

\begin{figure}[ht]
\centerline{\includegraphics[width=0.9\columnwidth, trim={20 40 40 30}]{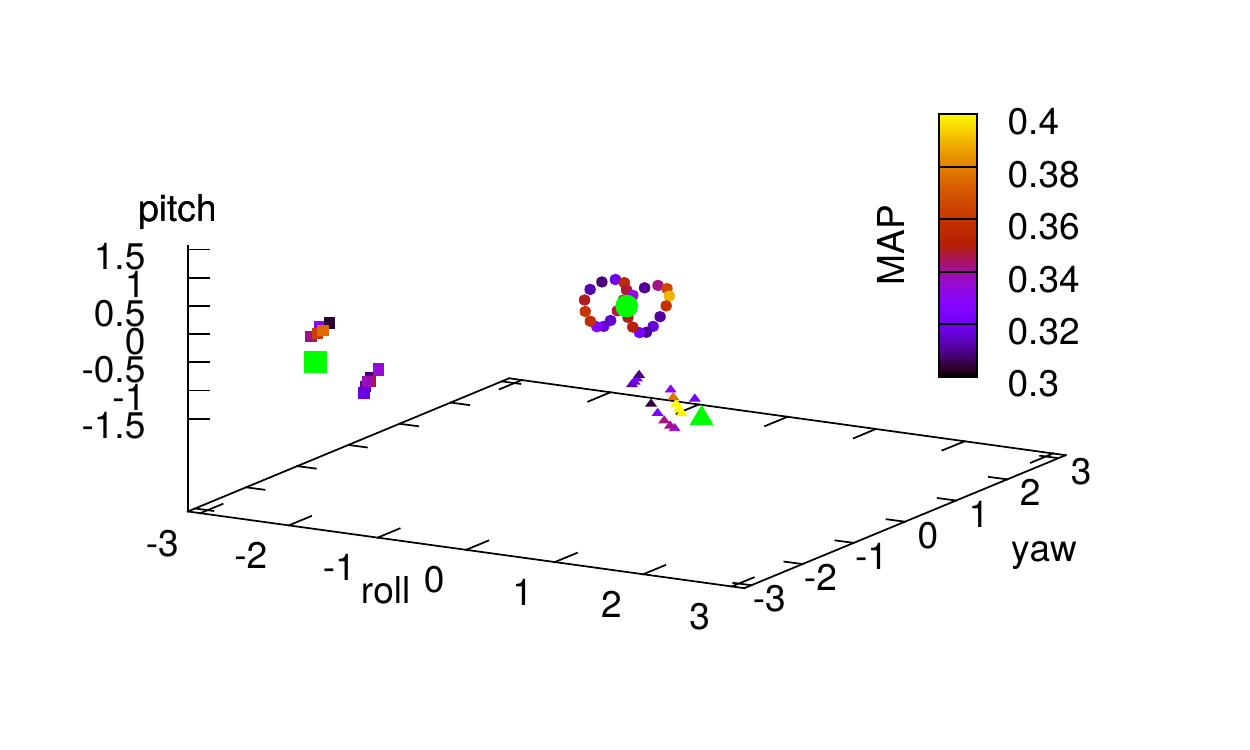}}
\caption{Tests with RDC and PCS. Cases $C_1$ (circles), $C_2$ (squares), $C_3$ (triangles).\label{fig:test3noerr}}
\end{figure}

In Figure~\ref{fig:test3noerr} we show the joint PCS+RDC case for determining an orientation. While the PCS do not resolve well the translation, they are useful to eliminate the ghost cones, thus determining the correct region of space for the orientations. The RDC and PCS formulas have the same type of symmetries, however the $E(P_j)$ vectors from (\ref{eq:eul1}) are different, so that in general the symmetries do not coincide.


\section{Tests with experimental error}
We added an uncorrelated Gaussian error to the mean measurements to take into account the experimental error. The error level was kept to $\pm{1}$ppm $\pm{10}\%$ for PCS and $\pm{1}$Hz $\pm{10}\%$ for RDC. We applied the algorithm of Subsection~\ref{ss:algorithm} and report the results in the following figures.

\begin{figure}[ht]
\centerline{\includegraphics[width=0.9\columnwidth, trim={20 40 40 30}]{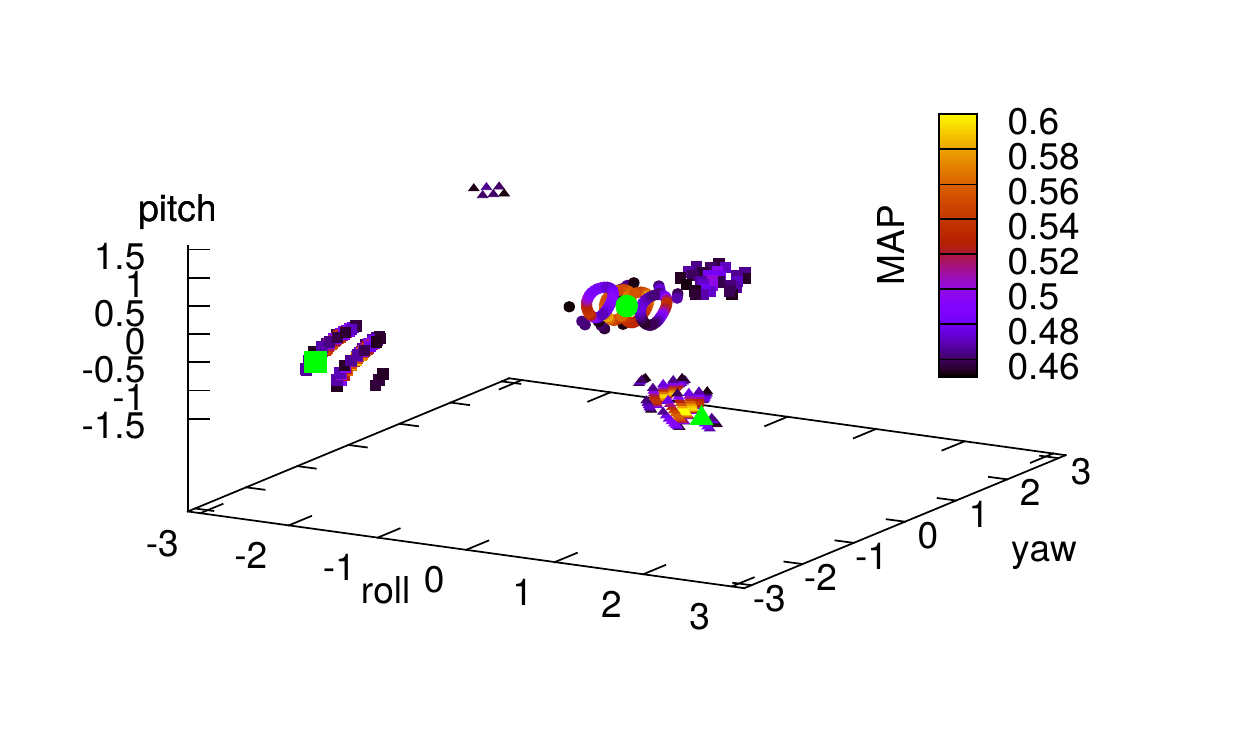}}
\caption{Tests with RDC alone. Cases $C_1$ (circles), $C_2$ (squares), $C_3$ (triangles).\label{fig:test1err}}
\end{figure}

In Figure~\ref{fig:test1err} we show the results of step 1 of the algorithm. Note that there are very few changes with respect to Figure~\ref{fig:test1noerr}. This is due to the fact that the mean RDC have only $15$ degrees of freedom, while we have $112$ measurements. The SVD algorithm implicitly fits the $15$ degrees of freedom of the mean RDC. Since the error is assumed to be Gaussian, a large number of measurements reduces the standard error of the fitted quantities. Hence the information of the RDC is well recovered even when the experimental error is considered. As explained in the previous section, coupling PCS and RDC helps removing ghost cones.

\begin{figure}[ht]
\centerline{
\includegraphics[width=0.49\columnwidth, trim={20 40 40 30}]{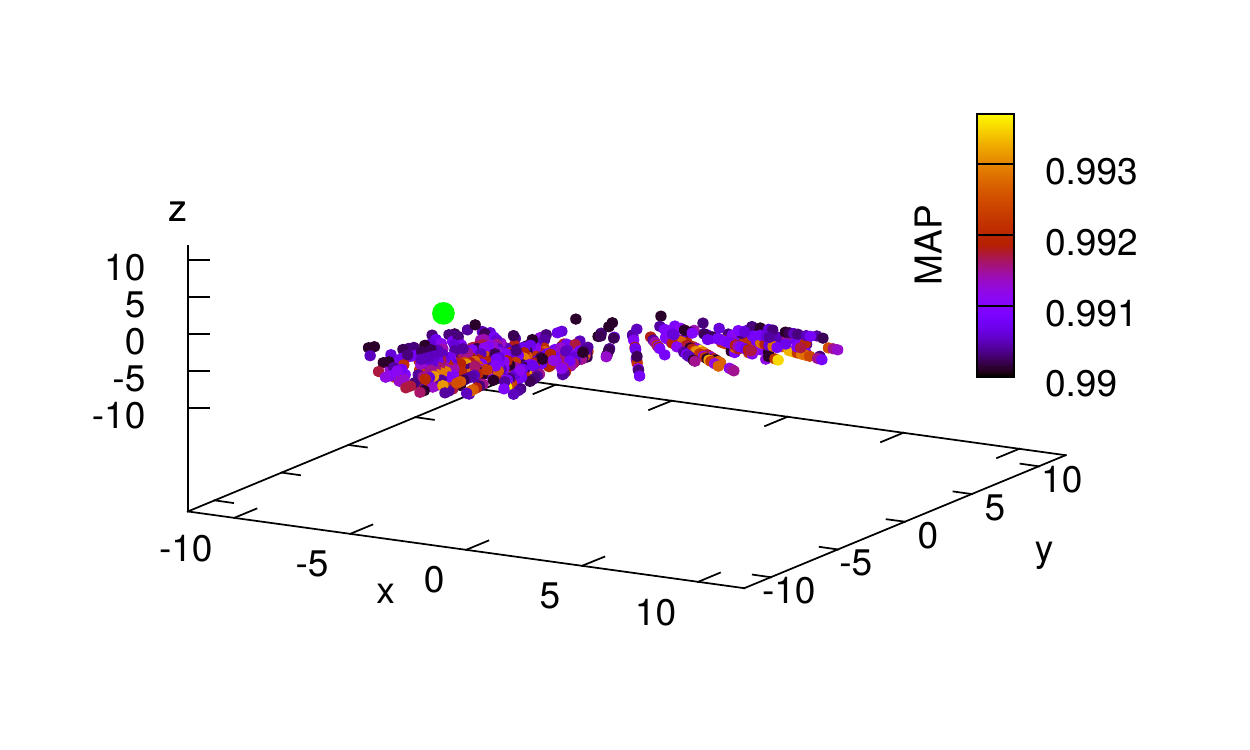}
\includegraphics[width=0.49\columnwidth, trim={20 40 40 30}]{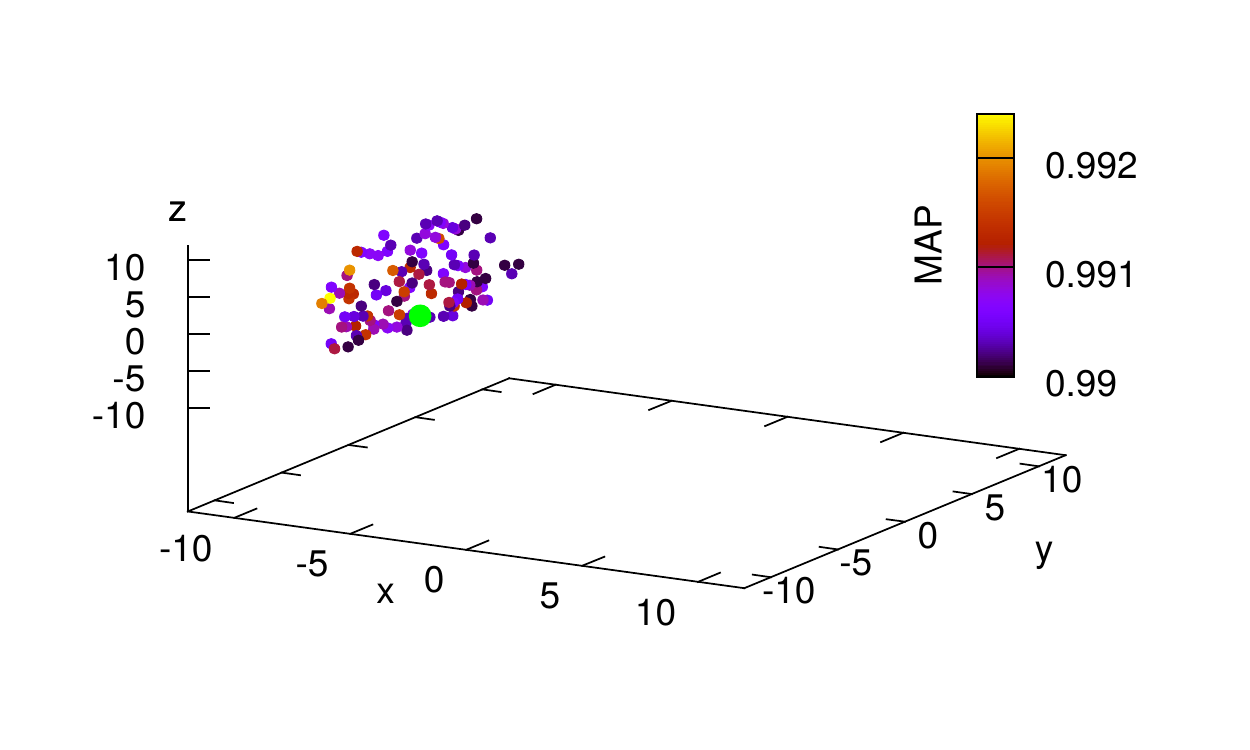}}
\caption{Tests with PCS alone. Cases $C_1$ (left panel) and $C_3$ (right panel).\label{fig:test2err}}
\end{figure}

In Figure~\ref{fig:test2err} we show the results of step 2 of the algorithm, in case $C_1$ (left panel) and $C_3$ (right panel). Here the introduction of the experimental error worsen the results. This should not be a surprise, since most of the mean PCS are under the $1$Hz threshold, so that their numerical value is destroyed by the error level which is larger. As already explained, the case $C_3$ is better because the C terminal is close to the metal.

\begin{figure}[ht]
\centerline{\includegraphics[width=0.5\columnwidth, trim={200 480 200 80}]{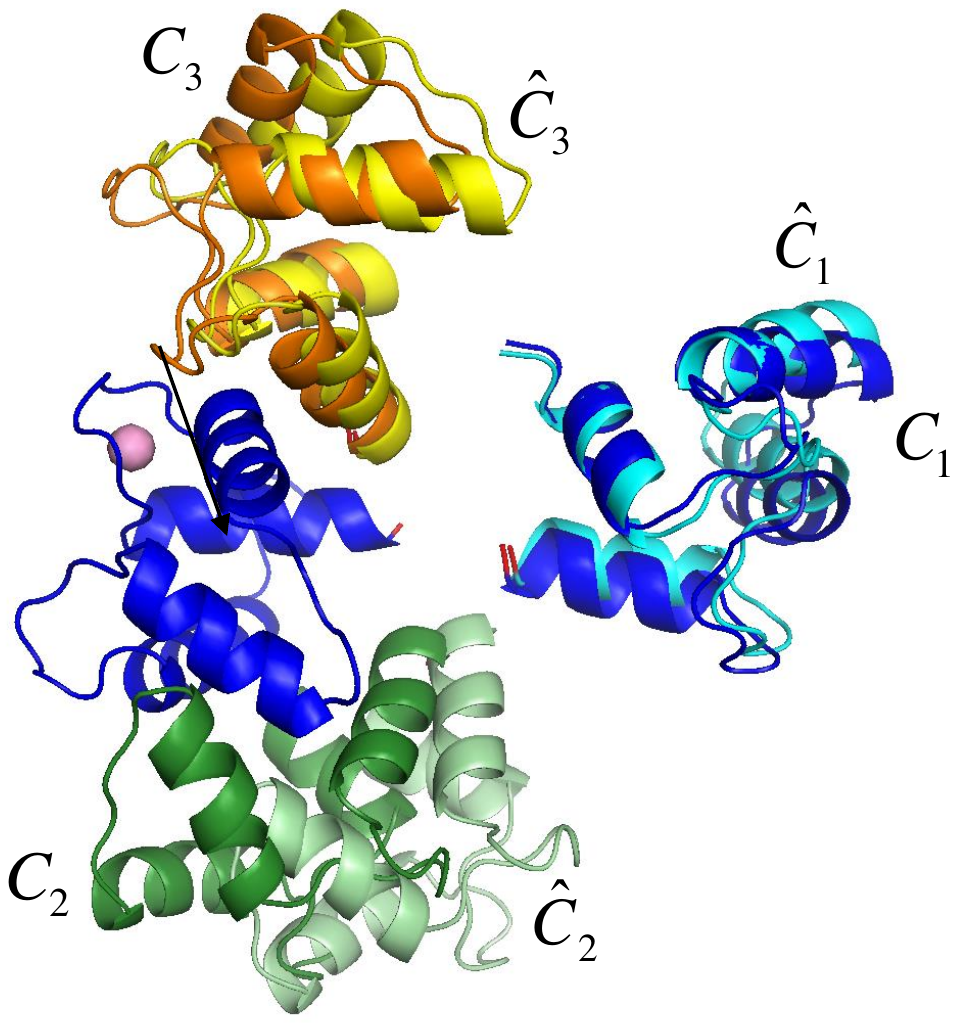}}
\caption{Tests with RDC and PCS. Metal \emph{(pink)}, N terminal \emph{(blue)}. Final reconstruction of the $C_1$, $C_2$ and $C_3$ cases. Dark colors show the positions with the largest probability, light colors show the reconstruction. \label{fig:res}}
\end{figure}

Figure~\ref{fig:res} shows the final reconstruction. Here the dark colored conformers show the positions with the largest probability, while the lightly colored conformers show the reconstructed positions $\hat C_k$.

It should be noted that the algorithm does not include a step where PCS and RDC are analysed together, if not in order to remove the ghost cones. We did attempt this step, using a local maximization technique. The results of this step are on average only marginally better than the results of the algorithm. If great care is not applied in the optimization, the error on the translation may even increase. Even in the joined RDC+PCS case, the translation is determined only by the PCS values, since there is no dependence on the RDC. However, the dependence on the translation is very weak, as it can be seen from figure~\ref{fig:test2err}, so that the variations in the MAP are largely due to the orientation of the conformer. More information on this additional attempt may be found in the supplementary material.

\section{Conclusions}
In this paper we demonstrated the ability of PCS and RDC measurements to recover a central tendency of an unknown probability distribution representing the conformational freedom of a protein made by two rigid domains connected by a flexible linker.

We made use of the MAP algorithm, extended to include PCS (and indeed any other class of measurements). The MAP algorithm determines the largest probability of a conformer in a distribution satisfying the measurements. Taken globally the MAP function is not a probability distribution but a sharp bound from above. For each position there is however an explicitly determined finite probability distribution with the MAP value as a weight for that conformer.

The RDC measurements are well able to determine any central tendency in the orientation of the conformer. Adding the PCS helps removing symmetric orientations, since the symmetries of the PCS are different from those of the RDC.

The identification of the translation is more difficult. The RDC does not depend on the translation, so we can only use PCS. However the information content of the PCS is very weak, and is further destroyed by the experimental noise. With exact data we can only approximatively determine the central tendency of the translation, see Figure~\ref{fig:test2noerr}. The situation worsen when the experimental error is added, as shown by Figure~\ref{fig:test2err}. Values of MAP larger than $0.9$ are obtained for a large fraction of the sample, and especially for positions relatively far from the metal. In other words, the conformer can sample any positions far from the metal for the $90\%$ of the time, resulting in very small partial values of the PCS. Adjusting the remaining $10\%$ of the distribution is enough to obtain the correct values of the PCS. As a consequence, the determination of the translation is less accurate.

We again stress that this is not due to the method employed. The MAP algorithm points out the extremal cases which should anyway be considered by any method trying to determine a solution. Or course there might be reasons to exclude some probability distributions, for instance using some threshold on the number of conformers or on the spread of the distribution. These are additional hypotheses which reduce the set of compatible solutions. However, since the problem is underdetermined and the real solution is not known, the reasons for removing these particular solutions should be soundly justified.

Other methods might not able to detect the MAP solutions, for instance if a predetermined sample is used for the conformers. Even if the sample is large, there might be a correlation between the rotational and translational part of the Euler transformations of the sample. In this case the identification of the orientation might bias the translation towards the correct value.

Further developments will include the case when there is more than a single central tendency in the data. Giving the results of this study, we foresee that additional information might be extracted only for the rotational part of the distribution. A possible way of overcoming this difficulty might be the inclusion of different measurements such as SAXS (Small angle X-Rays Scattering) \cite{SBK}. The NMR-SAXS integration has already been studied using the MO (Maximum Occurrence) method in \cite{BGLPPPRS}. Since the SAXS measurements depends on the global shape of the molecule this information might help in better determining the translational part of the Euler transformation.

\smallskip
\noindent{\bf Supplementary Information:} In the supporting material some additional results are presented. While not strictly necessary for the purpose of the paper, these results increase the insight on the behaviour of 
the PCS and RDC class of measurements.

\smallskip
\noindent{\it Acknowledgment:} As usual we wish to acknowledge the fruitful and long-established collaboration with the Center for Magnetic Resonance of the University of Florence. Discussions with applied scientists is always fruitful and helps gearing mathematics towards realistic problems.

\end{document}